\documentclass[12pt,reqno]{article}

\usepackage[usenames]{color}
\usepackage{amssymb}
\usepackage{graphicx}
\usepackage{amscd}
\usepackage{diagbox}
\usepackage{enumerate}

\usepackage[colorlinks=true,
linkcolor=webgreen,
filecolor=webbrown,
citecolor=webgreen]{hyperref}

\definecolor{webgreen}{rgb}{0,.5,0}
\definecolor{webbrown}{rgb}{.6,0,0}

\usepackage{color}
\usepackage{fullpage}
\usepackage{float}

\usepackage{graphics,amsmath,amssymb}
\usepackage{amsthm}
\usepackage{amsfonts}
\usepackage{latexsym}
\usepackage{epsf}
\usepackage{fullpage}

\title{New results on pseudosquare avoidance}

\author{Tim Ng\\
School of Computer Science\\
University of Waterloo\\
Waterloo, ON N2L 3G1 \\
Canada\\
\href{mailto:tim.ng@uwaterloo.ca}{\tt tim.ng@uwaterloo.ca} \\
\and
Pascal Ochem\\
LIRMM, CNRS \\
Universit\'e de Montpellier \\
France \\
\href{mailto:ochem@lirmm.fr}{\tt ochem@lirmm.fr} \\
\and
Narad Rampersad\\
Department of Math/Stats\\
University of Winnipeg\\
515 Portage Ave.\\
Winnipeg, MB R3B 2E9\\
Canada \\
\href{mailto:narad.rampersad@gmail.com}{\tt narad.rampersad@gmail.com} \\
\and
Jeffrey Shallit\\
School of Computer Science\\
University of Waterloo\\
Waterloo, ON N2L 3G1 \\
Canada\\
\href{mailto:shallit@uwaterloo.ca}{\tt shallit@uwaterloo.ca} 
}

\begin{document}

\theoremstyle{plain}
\newtheorem{theorem}{Theorem}
\newtheorem{corollary}[theorem]{Corollary}
\newtheorem{lemma}[theorem]{Lemma}
\newtheorem{proposition}[theorem]{Proposition}

\theoremstyle{definition}
\newtheorem{definition}[theorem]{Definition}
\newtheorem{example}[theorem]{Example}
\newtheorem{conjecture}[theorem]{Conjecture}

\theoremstyle{remark}
\newtheorem{remark}[theorem]{Remark}

\maketitle

\begin{abstract}
We start by considering binary words containing the minimum possible numbers of squares and antisquares (where an antisquare is a word of the form $x \overline{x}$), and we completely classify which possibilities can occur.  We consider avoiding $x p(x)$, where $p$ is any permutation of the underlying alphabet, and $x t(x)$, where $t$ is any transformation of the underlying alphabet.    Finally, we prove the existence of an infinite binary word simultaneously avoiding all occurrences of $x h(x)$ for {\it every\/} nonerasing morphism $h$ and all 
sufficiently large words $x$.   
\end{abstract}

\section{Introduction}

Let $x, v$ be words.  We say that $v$ is a {\it factor\/} of $x$ if there exist words $u, w$ such that $x = uvw$.  For example,
{\tt or} is a factor of {\tt word}.

By a {\it square\/} we mean a nonempty word of the form
$xx$, like the French word {\tt couscous}.
The {\it order\/} of a square $xx$ is $|x|$, the length of $x$.
It is easy to see that every binary word of length at least $4$
contains a square factor.  However, in a classic paper from 
combinatorics on words,
Entringer, Jackson, and Schatz \cite{Entringer&Jackson&Schatz:1974}
constructed an infinite binary word containing, as factors,
only $5$ distinct squares:  $0^2$, $1^2$, $(01)^2$,
$(10)^2$, and $(11)^2$.
This bound of $5$ squares was improved to $3$ by Fraenkel and Simpson \cite{Fraenkel&Simpson:1995}; it is optimal.
For some other constructions also achieving the bound $3$, see
\cite{Rampersad&Shallit&Wang:2005,Ochem:2006,Harju&Nowotka:2006,Badkobeh&Crochemore:2012}.

Instead of considering squares, one
could consider {\it antisquares}:  these are binary words of the form $x \overline{x}$,
where $\overline{x}$ is a coding that maps $0 \rightarrow 1$ and $1 \rightarrow 0$.  For example, {\tt 01101001} is an antisquare. (They should not be confused with the different notion of antipower recently introduced by
Fici, Restivo, Silva, and Zamboni \cite{Fici&Restivo&Silva&Zamboni:2018}.)
Clearly it is possible to construct an infinite binary word that avoids all antisquares, but only in a trivial way:  the only such words are $0^\omega = 000\cdots$ and $1^\omega = 111\cdots$.
Similarly, the only infinite binary words with exactly one antisquare are
$01^\omega$ and $10^\omega$.  However,
it is easy to see that
every word in 
$\{  1000, 10000 \}^\omega$ has exactly
two antisquares --- namely $01$ and $10$ --- and hence there are infinitely many such words that are aperiodic. 

Several writers have considered variations on these results.
For example, Blanchet-Sadri, Choi, and Merca\c{s} \cite{Blanchet-Sadri&Choi&Mercas:2011} considered avoiding large squares in partial words.  Chiniforooshan, Kari, and Zhu \cite{Chiniforooshan&Kari&Xu:2012} studied avoiding words of the form $x \theta(x)$, where $\theta$ is an antimorphic involution.
Their results implicitly suggest the general problem of simultaneously avoiding what we might call {\it pseudosquares}:  patterns of the form $x x'$, where $x'$ belongs to some (possibly infinite) class of modifications of $x$.

This paper has two goals.  First, for all integers $a, b \geq 0$ we determine
whether there is an infinite binary word having at most $a$ squares and $b$ antisquares.  If this is not possible, we determine the length of the longest finite binary word with this property.

Second, we apply our results to discuss the simultaneous avoidance of $x x'$, where $x'$ belongs to some class of modifications of $x$.  We consider three cases:  
\begin{itemize}
    \item[(a)] where $x' = p(x)$ for a permutation $p$ of the underlying alphabet; 
    \item[(b)] where $x' = t(x)$ for a transformation $t$ of the underlying alphabet; and 
    \item[(c)] where $x' = h(x)$ for an arbitrary nonerasing morphism.
\end{itemize}
In particular, we prove the existence of an infinite binary word that avoids $x h(x)$ simultaneously for all nonerasing morphisms $h$ and all sufficiently long words $x$.

\section{Simultaneous avoidance of squares and antisquares}

We are interested in binary words where the number of distinct factors that are squares and antisquares is bounded.   More specifically, we completely solve this problem determining in every case the length of the longest word having at most $a$ distinct squares and at most $b$ distinct antisquares.
Our results are summarized in the following table.  If (one-sided) infinite words are possible, this is denoted by writing $\infty$ for the length.

\begin{figure}[H]
\begin{center}
    \begin{tabular}{|c|ccccccccccccccc|}
    \hline
    \diagbox{$a$}{$b$} & 0 & 1 & 2 & 3 & 4 & 5 & 6 & 7 & 8 & 9 & 10 & 11 & 12 & 13  & $\cdots$\\
    \hline
    0 & 1 & 2 & 3 & 3 & 3 & 3 & 3 & 3 & 3 & 3 & 3 & 3 & 3 & 3 & $\cdots$ \\
    1 & 3 & 4 & 7 & 7 & 7 & 7 & 7 & 7 & 7 & 7 & 7 & 7 & 7 & 7 & $\cdots$\\
    2 & 5 & 6 & 11 & 11 & 11 & 11 & 12 & 12 & 12 & 13 & 15 & 18 & 18 & 18 & $\cdots$ \\
    3 & 7 & 8 & 15 & 15 & 15 & 20 & 20 & 20 & 24 & 29 & 34 & 53 & 98 & $\infty$ & $\cdots$ \\
    4 & 9 & 10 & 19 & 19 & 27 & 31 & 45 & 56 & 233 & $\infty$ & $\infty$ & $\infty$ & $\infty$ & $\cdots$ & \\
    5 & 11 & 12 & 27 & 27 & 40 & $\infty$ & $\infty$ & $\infty$ & $\infty$ & $\cdots$ & & & & &\\
    6 & 13 & 14 & 35 & 38 & 313 & $\infty$ & $\cdots$ &&&&&&&&\\
    7 & 15 & 16 & 45 & $\infty$ & $\infty$ & $\cdots $ &&&&&&&&&\\
    8 & 17 & 18 & 147 & $\infty$ & $\cdots $ &&&&&&&&&& \\
    9 & 19 & 20 & $\infty$ & $\cdots$ &&&&&&&&&&&\\
    10 & 21 & 22 & $\infty$ & $\cdots$ &&&&&&&&&&&\\
    $\vdots $ & & & & & & & & & & & & & & & \\
    \hline
    \end{tabular}
    
\end{center}
\caption{Length of longest binary word having at most $a$ squares and $b$ antisquares}
\end{figure}

The results in the first two columns and first three rows (that is, for $a \leq 2$ and $b \leq 1$) are very easy.
We first explain the first two columns:
\begin{proposition}
\leavevmode
\begin{itemize}
    \item[(a)]  For $a \geq 0$, the longest binary word with $a$ squares and $0$ antisquares has length $2a+1$.
    \item[(b)]  For $a \geq 0$, the longest binary word with $a$ squares and $1$ antisquare has length $2a+2$.
\end{itemize}
\end{proposition}

\begin{proof}
\leavevmode
\begin{itemize}
    \item[(a)]  If a binary word has no antisquares, then in particular it has no occurrences of either $01$ or $10$.
    Thus it must contain only one type of letter.  If it has length $2a+2$, then it has $a+1$ squares, of order $1,2,\ldots,a+1$.  If it has length $2a+1$, it has $a$ squares.  So $2a+1$ is optimal.
    
    \item[(b)]  If a length-$n$ binary word $w$ has only one antisquare, this antisquare must be either $01$ or $10$; without loss of generality, assume it is $01$.  Then $w$ is either of the form $0^{n-1} 1$ or $0 1^{n-1}$.  Such a word clearly has $\lfloor (n-1)/2 \rfloor$ squares.
    \end{itemize}
\end{proof}

We next explain the first three rows:  if a binary word has no squares, its length is clearly bounded by $3$, as we remarked earlier.  If it has one square, a simple argument shows it has length at most $7$. Finally, if it has two squares, already Entringer, Jackson, and Schatz \cite[Thm.~2]{Entringer&Jackson&Schatz:1974} observed that it has length at most $18$.

For all the remaining finite entries, we obtained the result through the usual backtrack search method, and we omit the details.

It now remains to prove the results labeled $\infty$.  First, we introduce some morphisms.  Let the morphisms
$h_{3,13}$,
$h'_{3,13}$,
$h_{4,9}$,
$h_{5,5}$,
$h_{7,3}$,
$h'_{7,3}$
be defined as follows:

\begin{itemize}
    \item[(a)]  $h_{3,13}: \\
    0 \rightarrow$ {\tiny\rm 001011001110001100101110001011000111001011001110001100101100011100101110001011000111001011001110001100101110001011}$\\
    ${\tiny\rm 001110001100101100011100101110001011001110001100101110001011000111001011100010110011100011001011000111}\\
$ 1\rightarrow$ {\tiny\rm 00101100111000110010111000101100011100101100111000110010110001110010111000101100011100101100111000110010111000101100}$\\
${\tiny\rm 0111001011100010110011100011001011100010110001110010110011100011001011100010110011100011001011000111} \\
$2 \rightarrow$ {\tiny\rm 
00101100111000110010111000101100011100101100111000110010110001110010111000101100011100101100111000110010111000101100}$\\
${\tiny\rm 0111001011100010110011100011001011000111001011001110001100101110001011001110001100101100011100101110}\\
This is a $216$-uniform morphism.

\item[(b)]
$h'_{3,13}: \\
0 \rightarrow 0010110011100011$ \\
$1 \rightarrow 001011000111$ \\
$2 \rightarrow 00101110$

\item[(c)] 
$h_{4,9}: \\
0 \rightarrow 0000101110000011000010110000011000101100001011100010110$\\
$1 \rightarrow 0000101110000011000010110000011000101100000101110001011$\\
$2 \rightarrow 0000101110000011000010110000010111000101100000110001011$\\
This is a $55$-uniform morphism.

\item[(d)] 
$ h_{5,5}: \\ 0 \rightarrow
101000001011000010100001101011000001$\\
$ 1\rightarrow 101000001011000001101011000010100001$\\
$2 \rightarrow 101000001010000110000010100000110000$\\
This is a $36$-uniform morphism.

\item[(e)] 
$h_{7,3}: \\  0 \rightarrow 00101000010010010100000101001$\\
$1 \rightarrow 00101000010010010000101001000$\\
$2 \rightarrow 00101000010010010000101000001$\\
This is a $29$-uniform morphism.

\item[(f)]
$h'_{7,3}:\\ 0 \rightarrow 0100100100001010000$ \\
$1 \rightarrow 01001001000001$ \\
$2 \rightarrow 0100100101000$\\

\item[(g)] 
$h_{9,2}: \\  0 \rightarrow 0001000100000001000101$\\
$1 \rightarrow 0000010001000100000101$\\
$2 \rightarrow 0000001000100000010100 $\\
This is a $22$-uniform morphism.

\end{itemize}
\begin{theorem}
Let $\bf w$ be an infinite squarefree sequence over the alphabet $\{0,1,2\}$.  Then 
\begin{itemize}
    \item[(a)]  $h_{3,13}({\bf w})$
 has $3$ squares and $13$ antisquares.  The squares are
 $0^2$, $1^2$, and $(01)^2$.  The antisquares are
 $01$, $10$, $0011$, $0110$, $1001$, $1100$,
 $000111$, $001110$, $011100$, $100011$, $110001$, $111000$, and $10010110$.

\item[(b)] $h'_{3,13}({\bf w})$ has $3$ squares and $13$ antisquares.  The squares are
 $0^2$, $1^2$, and $(01)^2$.  The antisquares are
 $01$, $10$, $0011$, $0110$, $1001$, $1100$,
 $000111$, $001110$, $011100$, $100011$, $110001$, $111000$, and $10010110$.
 
\item[(c)] 
$h_{4,9}({\bf w})$ 
has exactly $4$ squares and $9$ antisquares.  The squares are $0^2$, $1^2$, $(00)^2$, and $(01)^2$, and
the antisquares are $01$,   $10$,  $0011$, $0110$, $1100$, $011100$, $110001$, $111000$, and $1110000011$.

\item[(d)] 
$ h_{5,5}({\bf w})$
has exactly $5$ squares and $5$ antisquares.   The squares are $0^2$, $1^2$, $(00)^2$, $(01)^2$, and $(10)^2$, and
the antisquares are $01, 10, 0011, 0110,$ and $1100$.

\item[(e)] 
$h_{7,3}({\bf w})$
has $7$ squares and $3$ antisquares.    The squares are $0^2$, $(00)^2$, $(01)^2$, $(10)^2$, $(001)^2$, $(010)^2$, and
$(100)^2$, and the antisquares are
$01$, $10$, and $1001$.

\item[(f)] 
$h'_{7,3}({\bf w})$
has $7$ squares and $3$ antisquares.    The squares are $0^2$, $(00)^2$, $(01)^2$, $(10)^2$, $(001)^2$, $(010)^2$, and
$(100)^2$, and the antisquares are
$01$, $10$, and $1001$.

\item[(g)] 
$h_{9,2}({\bf w})$
has $9$ squares and $2$ antisquares.  The squares are
$0^2$, $(00)^2$, $(01)^2$, $(10)^2$, $(000)^2$, $(0001)^2$, $(0010)^2$, $(0100)^2$, and $(1000)^2$,
and the antisquares are $01$ and $10$.

\end{itemize}
\label{big}
\end{theorem}

\begin{proof}
Let $h$ be any of the morphisms above.
We first show that large squares are avoided.
The $h$-images of the letters have been ordered such that $|h(0)|\ge|h(1)|\ge|h(2)|$.
A computer check shows that for every letter $i$ and every ternary word $w$, the factor $h(i)$ appears in $h(w)$ 
only as the $h$-image of $i$.
Another computer check shows that for every ternary squarefree word $w$, the only squares $uu$ with $|u|\le 2|h(0)|-2$
that appear in $h(w)$ are the ones we claim.
If $w$ contains a square $uu$ with $|u|\ge 2|h(0)|-1$, then $u$ contains the full $h$-image of some letter.
Thus, $uu$ is a factor of $h(avbvc)$ with $a,b,c$ single letters and $v$ a nonempty word.
Moreover, $a\ne b$ and $b\ne c$, since otherwise $avbvc$ would contain a square.
It follows that $u=ph(v)s$, so that $p$ is a suffix of $h(a)$, and $h(b)=sp$, and $s$ is a prefix of $h(c)$.
Thus, $h(abc)$ contains the square $psps$ with period $|ps|$ at least $|h(2)|/2+1$ and at most $3|h(0)|/2$.
This contradicts our computer check, which rules out squares with period at least 5 and at most $2|h(0)|-2$.

To show that large antisquares are avoided, it suffices to exhibit a factor $f$ such that $f$ is uniformly recurrent in $h(w)$
and $\overline{f}$ is not a factor of $h(w)$. We use $f=0101$ for $h_{3,13}$ and $f=0^4$ for the other morphisms. 
\end{proof}

\begin{remark}  
The uniform morphisms were found as follows:  
for increasing values of $q$, our program looks (by backtracking) for a binary word of length $3q$ corresponding to the image $h(012)$ of $012$ by a suitable $q$-uniform morphism $h$. Given a candidate $h$,
we check that $h(w)$ has at most $a$ squares and $b$ antisquares
for every squarefree word $w$ up to some length.
Standard optimizations are applied to the backtracking.
Squares and antisquares are counted naively
(recomputed from scratch at every step), which is sufficient
since the morphisms found are not too large.
\end{remark}

\begin{remark}
The reader can check that $h_{3,13} = h'_{3,13} \circ m$, where $m$ is the $18$-uniform morphism given by
\begin{align*}
0 &\rightarrow 021012102012021201 \\
1 &\rightarrow  021012102120210201 \\
2 &\rightarrow 021012102120102012 \ .
\end{align*}
\end{remark}

\begin{corollary}
There exists an infinite binary word having at most ten distinct squares and antisquares as factors, but the longest binary word having nine or fewer distinct squares and antisquares is of length $45$.
\end{corollary}

\begin{remark}
A word of length $45$ with nine distinct squares and antisquares is
$$000001000000010100000010000101000000010000101.$$
\end{remark}

\begin{corollary}
Every infinite word having at most ten distinct squares and antisquares has critical exponent at least
$5$, and there is such a word having $5$-powers but no powers of higher exponent.
\end{corollary}

\begin{proof}
By the usual backtracking approach, we can easily verify that the longest finite
word having at most ten distinct antisquares, and critical exponent $< 5$ is of
length $57$.   One such example is
$$010001010000100100100001010010010100001001001000010100010.$$
On the other hand, if $\bf w$ is any squarefree ternary infinite word, then
from above  we know that the only possible squares that
can occur in $h_{5,5} ({\bf w})$ are of the form
$x^2$ for $x \in \{ 0, 1, 00, 01, 10 \}$.  
It is now easy to verify that the largest power of $0$ that occurs in  $h_{5,5} ({\bf w})$
is $0^5$; the largest power of $1$ that occurs is $1^2$; the largest power
of $01$ that occurs is $(01)^{5/2}$; and the largest power of
$10$ that occurs is $(10)^{5/2}$.
\end{proof}

\section{Pseudosquare avoidance}

In this section we discuss avoiding $x x'$ where $x'$ belongs to some large class of modifications of $x'$.  This is in the spirit of previous results \cite{Rumyantsev&Ushakov:2006,Durand&Levin&Shen:2008,Miller:2012}, where one is interested in avoiding factors of low Kolmogorov complexity.   The problems we study are not quite so general, but our results are effective, and we obtain explicit bounds.

\subsection{Avoiding pseudosquares for permutations}

Here we are interested in avoiding patterns of the form $x p(x)$, for
{\it all} codings $p$ that are permutations of the underlying alphabet.  Of course, this is impossible for words of length $\geq 2$  strictly  as  stated,  since  every  word  of  length  2  is  of  the  form $a p(a)$ where $p$ is  the permutation sending the letter $a$ to $p(a)$.  Thus
it is reasonable to ask about avoiding $x p(x)$ for all
words $x$ of length $\geq n$.  Our first result shows this
is impossible for $n = 2$.

\begin{theorem}
For all finite alphabets $\Sigma$,
and for all words $w$ of length $\geq 10$ over $\Sigma$,
there exists a permutation $p$ of $\Sigma$ and
a factor of $w$ of the form $x x'$, where $x' = p(x)$,
and $|x| \geq 2$.
\end{theorem}

\begin{proof}
Using the usual tree-traversal technique, where we extend the alphabet
size at each length extension.
\end{proof}

We now turn to the case of larger $n$.  For $n \geq 3$, and $k = 2$, we can avoid all factors of
the form $x p(x)$.
Of course, this case is particularly simple, since there are only two permutations of the alphabet:
the identity permutation that leaves letters invariant, and the map $x \rightarrow 
\overline{x}$, which changes $0$ to $1$ and vice versa.

\begin{theorem}
There exists an infinite word ${\bf w}$ over the binary alphabet
$\Sigma_2 = \{ 0, 1 \}$ that avoids $xx$ and $x\overline{x}$ for all
$x$ with $|x| \geq 3$.  
\label{two}
\end{theorem}

\begin{proof}
We can use the morphism in Theorem~\ref{big} (c).  Alternatively, a simpler proof comes from
the fixed point of the morphism
\begin{align*}
0 &\rightarrow 01\\
1 &\rightarrow 23\\
2 &\rightarrow 24\\
3 &\rightarrow 51\\
4 &\rightarrow 06\\
5 &\rightarrow 01\\
6 &\rightarrow 74 \\
7 &\rightarrow 24
\end{align*}
followed by the coding $n \rightarrow n \bmod 2$.
We can now use {\tt Walnut} \cite{Mousavi:2016} to verify that the 
resulting $2$-automatic word has the desired property.
This word has exactly
$5$ distinct squares:  
$$0^2, 1^2, (00)^2, (01)^2, (10)^2,$$
and exactly $6$ distinct antisquares:
$$ 01, 10, 0011, 0110, 1001, 1100.$$
\end{proof}

\subsection{Avoiding pseudosquares for transformations}

In the previous subsection we considered permutations of the
alphabet.  We now generalize this to {\it transformations} of 
the alphabet, or, in other words, to arbitrary codings (letter-to-letter
morphisms).

\begin{theorem}
\leavevmode
\begin{itemize}
    \item[(a)]  For all finite alphabets $\Sigma$, and all words $w$ of length $\geq 31$ over $\Sigma$, there exists a transformation $t:\Sigma^* \rightarrow \Sigma^*$ such that $w$ contains a factor of the form $x t(x)$ for $|x| \geq 3$.

\item[(b)]  For all finite alphabets $\Sigma$,
and all words $w$ of length $\geq 16$ over $\Sigma$, there exists
a transformation $t$ of $\Sigma$ such that $w$ contains 
a factor of the form $x x'$, where $x' = t(x)$ or $x = t(x')$
and $|x| \geq 3$.
\end{itemize}
\end{theorem}

\begin{proof}
Using the usual tree-traversal technique, where we extend the alphabet
size at each length extension.
\end{proof}

We now specialize to the binary alphabet.   This case is particularly simple, since in addition to the
two permutations of the alphabet, the only other transformations are the ones sending both
$0,1$ to a single letter (either $0$ or $1$).

\begin{theorem}
There exists an infinite word ${\bf w}$ over the binary alphabet
$\Sigma_2 = \{ 0, 1 \}$ avoiding $0^4$, $1^4$, and $xx$ and $x \overline{x}$ 
for every $x$ with $|x| \geq 4$.
In other words, $\bf w$ avoids both
$x t(x)$ and $t(x) x$ for
$|x| \geq 4$ and all transformations $t$.
There is no such infinite word if $4$ is changed to $3$.
\end{theorem}

\begin{proof}
Use the fixed point of the morphism
\begin{align*}
0 &\rightarrow 01\\
1 &\rightarrow 23\\
2 &\rightarrow 45\\
3 &\rightarrow 21\\
4 &\rightarrow 23\\
5 &\rightarrow 42\\
\end{align*}
followed by the coding $n \rightarrow \lfloor n/3 \rfloor$.  The result can now easily be verified with {\tt Walnut}.
\end{proof}

\subsection{Avoiding pseudosquares with morphic images}

In this subsection we consider simultaneously avoiding all patterns of the
form $x h(x)$, for all morphisms $h$ defined 
over $\Sigma_k = \{ 0,1,\ldots, k-1\}$.  Clearly this is impossible if $h$ is
allowed to be erasing (that is, some images are allowed
to be empty), or if $x$ consists of a single letter.
So once again we consider the question for sufficiently
long $x$.

For this version of the problem, it is particularly hard to obtain experimental data, because 
the problem of determining, given $x$ and $y$,
whether there is a morphism $h$ such that
$y = h(x)$, is NP-complete \cite{Angluin:1980,Ehrenfeucht&Rozenberg:1979}.

\begin{theorem}\label{thm9dot5}
No infinite binary word avoids all factors of the form $xh(x)$,
for all nonerasing binary morphisms $h$, with $|x|\ge4$.
\end{theorem}

\begin{proof}
This can be checked by computer in less than a second.
We give another proof that is a reduction to a more classical question of avoiding large squares and a finite set of factors.

Let $\bf w$ be a potential counter-example to Theorem~\ref{thm9dot5}.
Without loss of generality, we can assume that $\bf w$ is uniformly recurrent (see, e.g., \cite[Lemma 2.4]{deLuca&Varricchio:1991a}).
Suppose, to get a contradiction, that $\bf w$ contains the factor $000$.
Since ${\bf w} \ne 0^{\omega}$, the word $\bf w$ contains $1000$.
Since $\bf w$ is uniformly recurrent,
the factor $1000$ extends to a factor $1000u000$,
where $u$ is a nonempty finite word,
which is a forbidden occurrence of $xh(x)$.
So $\bf w$ avoids $000$, and by symmetry, the word $\bf w$ also avoids $111$.
Suppose, to get a contradiction, that $\bf w$ contains both $0100$ and $1011$.
The factor $0100$ extends to $01001$.
Since $\bf w$ is uniformly recurrent and contains $11$, the
word $\bf w$ contains $01001u11$, where $u$ is a nonempty finite word,
which is a forbidden occurrence of $xh(x)$.
So $\bf w$ does not contain both $0100$ and $1011$,
and we assume without loss of generality that $\bf w$ avoids $0100$.

Using the usual tree-traversal technique,
we can now easily check that no infinite binary word avoids
$000$, $111$, $0100$, and every square $xx$ with $|x|\ge4$.
Thus, $\bf w$ does not exist. 
\end{proof}

\begin{theorem}\label{thm10prime}
There exists an infinite binary word that avoids all factors of the form $xh(x)$,
for all nonerasing binary morphisms $h$, with $|x|\ge5$.
\end{theorem}

\begin{proof}
Let $\bf u$ be any infinite ternary $(7/4^+)$-free word, and consider the binary word $\bf w$ defined
by ${\bf w} = m({\bf u})$, where $m$ is the 246-uniform morphism given below.\\

\noindent $0\to$ {\tiny\rm 000110100110001110100101100011100101100111011010011100011010010110001110100110001110011010011100010110011101101001110001} \\
{\tiny\rm 101001100011100110100111000110100101100011101001100011100101100111011010011100011010011000111001101001110001011001110110100111 } \\
\noindent $1\to$ {\tiny\rm 000110100110001110100101100011100101100111011010011100010110011100011001011000111001011001110110100111000110100101100011} \\
{\tiny\rm 101001100011100101100111011010011100011010011000111001101001110001101001011000111010011000111001101001110001011001110110100111 } \\
\noindent $2\to$ {\tiny\rm 000110100101100011101001100011100101100111011010011100010110011100011001011000111001011001110110100111000110100101100011} \\
{\tiny \rm 101001100011100101100111011010011100011010011000111001101001110001011001110110100111000110100101100011101001100011100110100111 } \\

We use $a$ and $b$ to denote letters.  We will use the concept of generalized repetition threshold \cite{Ilie&Ochem&Shallit:2005}.
Recall that a word is said to be
$(e,n)$-free if it contains no factor of the form $x^f$ where $f \geq e$ and $|x| \geq n$.
We will need the following properties of $\bf w$.  
\begin{enumerate}[(a)]
 \item $\bf w$ is $(11/6^+,4)$-free. In particular, the only squares occurring in $\bf w$ are $00$, $11$, $0101$, $1010$, $010010$, $101101$, and $110110$.
 \item The only cubes occurring in $\bf w$ are $000$ and $111$. Every cube $bbb$ extends to the left to $\overline{bb}bbb$.
 \item $\bf w$ does not contain any of the following factors: $01010$, $10101$, $00100$, $1101100$, $1011010010$.
 \item Every factor of $\bf w$ of length 17 contains $00111$ or $11000$.
 \item Every factor of $\bf w$ of length 98 contains $11011$.
 \item Every factor of $\bf w$ of length at least 5, except $00010$, $11101$, $111011$, and $11011$,
 contains a factor of the form $bb\overline{bb}$, $b\overline{b}b\overline{b}$, $\overline{b}bb\overline{b}$, or $\overline{b}bbb\overline{b}$.
\end{enumerate}

By \cite[Lemma 2.1]{Ochem:2006}, it is sufficient to check the $(11/6^+,4)$-freeness for the image of every
$(7/4^+)$-free ternary word of length smaller than $\frac{2\times11/6}{11/6-7/4}=44$.
The other properties can be checked by inspecting factors of $w$ with bounded length.

The following cases show that $\bf w$ contains no factor of the form $xh(x)$ with $|x|\ge5$.
\begin{itemize}
 \item We can rule out $h(0)=h(1)$, as $h(x)$ contains $h(0)^5$, which contradicts (a).
 \item We can rule out $h(b)=b$, as $xh(x)=xx$ is a square with period at least $5$, which contradicts (a).
 \item We can rule out $|x|\ge17$: By (d), $x$ contains the factor $bb\overline{bbb}$.
 By (b), $|h(\overline{b})|=1$, say $h(\overline{b})=a$.
 By (a), a square of period at least two has either $a$ or $a\overline{a}$ as a suffix.
 So if $|h(b)|>1$, then $h(bb\overline{bbb})$ has either $aaaa$ or $a\overline{a}aaa$ as a suffix, which contradicts (b).
 Thus $|h(b)|=|h(\overline{b})|=1$.
 By the previous cases, the only remaining possibility is $h(b)=\overline{b}$.
 If $|x|\ge98$, then $x$ contains $11011$ by (e). Thus $h(x)$ contains $00100$, which contradicts (c).
 If $17\le|x|\le97$, then a computer check shows that $w$ contains no antisquare $xh(x)=x\overline{x}$.
 \item If $x=bbb\overline{b}b$, then $xh(x)$ contains the factor $\overline{b}bh(bbb)$, which contradicts (b).
 \item If $x=111011$ or $x=11011$, then $xh(x)$ contains the factor $11011h(11)$.
 We can check that every choice of $h(11)$ leads to a contradiction with (a), (b), or (c).
 \item We can rule out the remaining cases. By (f) and the previous two cases, we can assume that $x$
 contains $bb\overline{bb}$, $b\overline{b}b\overline{b}$, $\overline{b}bb\overline{b}$, or $\overline{b}bbb\overline{b}$.
 Since $b$ is always contained in a square, $|h(b)|\le3$ by (a).
 If $\overline{b}$ is contained in a square, then $|h(\overline{b})|\le 3$.
 Otherwise, $x$ contains $\overline{b}bb\overline{b}$, or $\overline{b}bbb\overline{b}$. Notice that $|h(bb)|\le6$ and $|h(bbb)|=3$.
 Let $s\in\{2,3\}$. The repetition $h(\overline{b}b^s\overline{b})$ in a $(11/6^+,4)$-free word implies that $|h(\overline{b}b^s\overline{b})|\le\tfrac{11}6|h(\overline{b}b^s)|$.
 This gives $|h(\overline{b})|\le5|h(b^s)|\le30$. Thus, $|h(0)|+|h(1)|\le 3+30=33$.
 So if $w$ contains a factor of the form $xh(x)$ with $|x|\ge5$, then $|x|\le16$ and $|h(0)|+|h(1)|\le33$.
 Finally, a computer check shows that $w$ contains no such factor $xh(x)$.
\end{itemize}
\end{proof}



\newcommand{\noopsort}[1]{} \newcommand{\singleletter}[1]{#1}

\end{document}